\documentclass{scrartcl}

\usepackage[english]{babel}

\usepackage{latexsym}
\usepackage{amsmath}
\usepackage{amsthm}
\usepackage{stmaryrd}
\usepackage[linesnumbered,boxed]{algorithm2e}
\usepackage{bookmark}
\usepackage{enumitem}

\SetAlCapSkip{1.5em}

\newtheorem{lemma}{Lemma}[section]
\newtheorem{theorem}[lemma]{Theorem}

\newtheorem{question}{Question}
\newtheorem{proposition}[lemma]{Proposition}

\newtheorem{definition}[lemma]{Definition}

\renewcommand{\phi}{\varphi}

\newcommand{\DR}{\mathrm{DR}}

\newcommand{\var}[1]{\mathrm{var}(#1)}
\newcommand\vek[1]{\mathbf{#1}}
\newcommand{\lit}[1]{\ensuremath{\operatorname{lit}(#1)}}

\newcommand{\litVar}[1]{\ensuremath{\llbracket #1\rrbracket}}
\newcommand{\cnflen}[1]{\ensuremath{\lVert #1\rVert}}
\newcommand\assign[1]{\mathbf{#1}}
\newcommand{\clsem}{\mathrm{cl}_\mathrm{sem}}
\newcommand{\clup}{\mathrm{cl}_\mathrm{up}}
\newcommand{\meta}{\mathrm{meta}}

\newcommand\mlabel[1]{\label{#1}}

\usepackage{array}
\newcounter{qhgrprowcntr}[table]
\renewcommand{\theqhgrprowcntr}{(q\arabic{qhgrprowcntr})}
\newcolumntype{Q}{>{\refstepcounter{qhgrprowcntr}{\theqhgrprowcntr}}c}

\usepackage{booktabs}

\usepackage{xcolor}
\definecolor{urlcolor}{HTML}{00009B}
\definecolor{citecolor}{HTML}{9B0000}

\hypersetup{
    colorlinks=true,       % false: boxed links; true: colored links
    linkcolor=urlcolor,        % color of internal links (change box color with linkbordercolor)
    citecolor=citecolor,         % color of links to bibliography
    filecolor=magenta,     % color of file links
    urlcolor=urlcolor,         % color of external links
    linktocpage=true
}

\begin{document}
\title{Bounds on the size of PC and URC formulas}
\author{Petr Ku{\v c}era\thanks{Department of Theoretical Computer Science
and Mathematical Logic,
Faculty of Mathematics and Physics,
Charles University, Czech Republic,
kucerap@ktiml.mff.cuni.cz}
   \and Petr Savick{\'y}\thanks{Institute of
      Computer Science of
   the Czech Academy of Sciences,
   Czech Republic,
   savicky@cs.cas.cz}
}
\date{\empty}
\maketitle

\begin{abstract}
In this paper we investigate CNF formulas, for which the
unit propagation is strong enough to derive a contradiction
if the formula together with a partial assignment of the
variables is unsatisfiable
(unit refutation complete or URC formulas)
or additionally to derive all implied literals
if the formula is satisfiable
(propagation complete or PC formulas).
If a formula represents a function using existentially
quantified auxiliary variables, it is called an encoding
of the function.
We prove several results on the sizes of PC and URC formulas
and encodings. One of them are separations
between the sizes of formulas of different types. Namely, we prove
an exponential separation between the size of URC formulas and
PC formulas and
between the size of PC encodings using auxiliary variables
and URC formulas. Besides of this, we prove
that the sizes of any two irredundant PC formulas for the same
function differ at most by a factor polynomial in the number of
the variables and
present an example of a function demonstrating that
a similar statement is not true for URC formulas.
One of the separations above implies that a q-Horn formula
may require an exponential
number of additional clauses to become a URC formula. On the other
hand, for every q-Horn formula, we present a polynomial size URC encoding
of the same function using auxiliary variables.
This encoding is not q-Horn in general.
\end{abstract}

% vim: set filetype=tex :
% vim: set spell spelllang=en :

\section{Introduction}

Since unit propagation is a basic procedure used in DPLL based SAT
solvers including CDCL solvers, it has become a common practice to require that unit
propagation maintains at least some level of local consistency in the
constraints being encoded into a conjunctive normal form (CNF) formula.
Close connection between unit propagation in SAT solvers and maintaining generalized arc
consistency was investigated for example in~\cite{B07}. General CNF
formulas do not support even efficient consistency testing, hence, it
is important to understand the properties needed for good propagation.
In this paper, we investigate CNF formulas with a high level of propagation strength --- they
are unit refutation complete or propagation complete. The class of unit
refutation complete (URC) formulas was introduced in~\cite{V94} and it
consists of formulas whose consistency with any given partial assignment can be
tested by unit propagation. It was shown in~\cite{GK13} that the class
of URC formulas coincides with the class SLUR (single lookahead unit
resolution) introduced in~\cite{AFSS95}. The class of propagation complete formulas
(PC) was later introduced in~\cite{BM12}.
PC formulas are URC and, moreover, unit propagation derives all implied
literals provided the formula is satisfiable with a given partial assignment
of its variables.

Every boolean function has a URC and a PC representation.
In particular, a canonical CNF which consists of all prime implicates is
always PC and, hence, URC\@. On the other hand, it is co-NP complete to check if a
formula is URC~\cite{CKV12} or PC~\cite{BBCGK13}.

Unit propagation is a weak
deduction rule, so PC or URC representations
of functions tend to be large. The power of the unit
propagation increases when auxiliary variables
are used and in this case, we call the corresponding
representations PC or URC encodings.
PC encodings were already used as a target language of knowledge
compilation. In particular, in~\cite{AGMES16} a construction of a PC
encoding of a function represented with an ordered decision diagram
was described. URC encodings were considered in~\cite{BJM12} as an
existential closure of URC formulas.

As one of the results of this paper, we provide a characterization of
propagation complete formulas using specific Horn functions. This
characterization applies also to encodings provided that the set of the auxiliary
variables is fixed. The importance of the use of auxiliary variables
is emphasized by exponential lower bounds for PC and URC formulas without
auxiliary variables for functions that have a polynomial size PC encodings.
Our lower bounds are stronger than previously known similar lower bounds
in that they apply to functions for which a polynomial size CNF representation
exists. Moreover, the functions belong to well-known tractable classes
of Horn and q-Horn formulas~\cite{BCH90}.

We also show that the sizes of irredundant PC formulas
representing the same function are polynomially related
which as we show is not the case of irredundant URC formulas.
The minimum size of a URC representation is always at most the
minimum size of a PC representation, however, although
irredundancy is a helpful criterion for PC formulas, it can be
completely misleading for URC ones. This result suggests that
the structure of PC formulas can make it easier to find
a PC formula of a reasonable size provided such a formula
exists rather than a URC formula.

Section~\ref{sec:basic} summarizes known notions used in this paper.
A complete summary of the results is presented in Section~\ref{sec:results}.
The proofs of the results are postponed to the remaining sections.

In Section~\ref{sec:lb}, we prove exponential lower bounds on the sizes of PC and URC
formulas for specific functions with a polynomial size CNF representation which
additionally admit a polynomial size PC encoding using auxiliary variables.
In Section~\ref{sec:irred-PC-URC}, we prove that on the contrary to URC formulas,
the sizes of irredundant PC formulas representing the same function differ at most
by a factor polynomial in the number of the variables.
In Section~\ref{sec:characterize-PC}, we present a characterization
of PC formulas based on a correspondence between a specific variant of
the dual rail encoding~\cite{BKNW09,BBIMM18,BJM12,BBBCS87,IMM17,MFSO97,MIBMB19}
of a PC formula and
a Horn function capturing semantical closure of the sets of the literals assuming
the given formula.
In Section~\ref{sec:q-horn:urc-enc}, we present a construction
of a URC encoding for a general q-Horn formula of size polynomial
in the size of the input formula. In a general case, the presented
encoding is not a q-Horn formula.
The paper is closed by formulating several directions for further research
in Section~\ref{sec:further-research}.

% vim: set filetype=tex :
% vim: set spell spelllang=en :

\section{Basic Concepts}
\label{sec:basic} % chktex 24

In this section, we recall the known concepts we use and introduce
the necessary notation.

\subsection{URC and PC Formulas}

A formula in conjunctive normal form (\emph{CNF formula}) is
a conjunction of clauses. A \emph{clause} is a disjunction of a set of
literals and a \emph{literal} is a variable \(x\) (\emph{positive literal}) or its negation
\(\neg x\) (\emph{negative literal}). Given a set of variables \(\vek{x}\), \(\lit{\vek{x}}\)
denotes the set of literals on variables in \(\vek{x}\).
We treat a clause as a set of literals and a CNF
formula as a set of clauses. In particular, \(|C|\) denotes the
number of literals in a clause \(C\) and \(|\varphi|\) denotes the
number of clauses in a CNF formula \(\varphi\). We denote
\(\cnflen{\varphi}=\sum_{C\in \varphi}|C|\) the \emph{length} of a CNF
formula \(\varphi\).

Clause \(C\) is \emph{Horn} if it contains at most one positive
literal, it is \emph{definite Horn}, if it contains exactly one
positive literal. A definite Horn clause \(\neg
   x_1\vee\cdots\vee\neg x_k\vee y\) represents the implication
\(x_1\land\cdots\land x_k\to y\) and we use both kinds of notation
interchangeably. The set of variables \(\{x_1, \ldots, x_k\}\) in the
assumption of a definite Horn clause is called its \emph{source set},
variable \(y\) is called its \emph{target}.

A \emph{partial assignment} \(\alpha\) of values to variables in \(\vek{z}\) is
a subset of \(\lit{\vek{z}}\) that does not contain a complementary
pair of literals, so we have \(|\alpha\cap\lit{x}|\le 1\) for each
\(x\in\vek{z}\). By \(\varphi(\alpha)\) we denote the formula
obtained from \(\varphi\) by the partial setting of the variables
defined by \(\alpha\). We identify a set of literals \(\alpha\) (in
particular a partial assignment) with the conjunction of these
literals if \(\alpha\) is used in a formula such as
\(\varphi(\vek{x})\land\alpha\).

We are interested in formulas which have good properties with respect
to unit propagation which is a well known procedure in SAT
solving~\cite{BHMW09}. For technical reasons, we represent
unit propagation using unit resolution.
The \emph{unit resolution} rule allows to
derive clause \(C\setminus\{l\}\) given a clause \(C\) and a unit
clause \(\neg l\). A clause \(C\) can be derived from
\(\varphi\) by \emph{unit resolution}, if \(C\) can be derived from
\(\varphi\) by a series of applications of the unit resolution rule
and we denote this fact with
\(\varphi\vdash_1C\). The notion of a propagation complete CNF formula
was introduced in~\cite{BM12} as a generalization of a unit refutation
complete CNF formula introduced in~\cite{V94}.

\begin{definition}\mlabel{def:urc-pc} % chktex 24
Let $\phi(\vek{x})$ be a CNF formula.
\begin{itemize}
\item We say that $\phi$ is \emph{unit refutation complete} (\emph{URC})
if the following implication holds for every partial assignment
$\alpha \subseteq \lit{\vek{x}}$
\begin{equation}
   \label{eq:def-urc}
\phi(\vek{x}) \wedge \alpha \models \bot
\hspace{1em} \Longrightarrow \hspace{1em} \phi \wedge \alpha \vdash_1 \bot\,.
\end{equation}
\item We say that $\varphi$ is \emph{propagation complete} (\emph{PC})
if for every partial assignment $\alpha \subseteq \lit{\vek{x}}$ and
for every $l \in \lit{\vek{x}}$, such that
\begin{equation} \label{eq:def-pc-1}
\phi(\vek{x}) \wedge \alpha \models l
\end{equation}
we have
\begin{equation} \label{eq:def-pc-2}
\phi \wedge \alpha \vdash_1 l
\hspace{1em} \text{or} \hspace{1em}
\phi \wedge \alpha \vdash_1 \bot\,.
\end{equation}
\end{itemize}
\end{definition}

One can verify that a formula $\phi(\vek{x})$ is PC
if and only if for every partial assignment $\alpha\subseteq\lit{\vek{x}}$
and for every \(l\in\lit{\vek{x}}\) such that
$$
\varphi(\vek{x})\wedge\alpha\not\vdash_1 \bot
\hspace{1em}\text{and}\hspace{1em}
\varphi(\vek{x})\wedge\alpha\not\vdash_1 \neg l
$$
the formula $\phi(\vek{x})\land\alpha\wedge l$ is satisfiable.

Given a boolean function \(f(\vek{x})\) on variables \(\vek{x}=(x_1,
   \dots, x_n)\), we differentiate between CNF formulas which
represent \(f(\vek{x})\) with using only variables \(\vek{x}\) and CNF
encodings which can posibly use existentially quantified auxiliary
variables in addition to the input variables \(\vek{x}\). Formally, we
define the CNF encodings as follows.

\begin{definition}[Encoding] \mlabel{def:cnf-enc} % chktex 24
Let \(f(\vek{x})\) be a boolean function on variables
\(\vek{x}=(x_1, \dots, x_n)\). Let \(\varphi(\vek{x},\vek{y})\)
be a CNF formula on \(n+m\) variables where
\(\vek{y}=(y_1, \dots, y_m)\).
We call \(\varphi\) a \emph{CNF encoding} of \(f\) if
for every \(\assign{a}\in{\{0, 1\}}^{\vek{x}}\) we have
\begin{equation} \label{eq:enc-def}
f(\assign{a}) = (\exists \assign{b}\in{\{0, 1\}}^{\vek{y}})\,
\varphi(\assign{a}, \assign{b})\,,
\end{equation}
where we identify $1$ and $0$ with logical values true and false.
The variables in \(\vek{x}\) and \(\vek{y}\) are called \emph{input variables}
and \emph{auxiliary variables}, respectively.
\end{definition}

The notions of PC and URC encodings are now defined as follows.

\begin{definition}
\emph{PC encoding} and \emph{URC encoding}
is an encoding that is a PC formula and a URC formula,
respectively.
\end{definition}

Finally, let us note that given a boolean function \(f(\vek{x})\), a
CNF \(\varphi(\vek{x})\) which consists of all prime implicates of
\(f(\vek{x})\) is always PC and also URC, thus every
boolean function has a PC formula which represents it.

\subsection{Implicational Dual Rail Encoding}
\label{sub:dual-rail:def}

We use the well-known dual rail encoding of partial
assignments~\cite{BBIMM18,BBBCS87,IMM17,MFSO97,MIBMB19} to simulate
unit propagation in a general CNF formula in the same way as in~\cite{BJM12,BKNW09}.
Let us describe a variant of such a simulation suitable for our purposes.

In short, if $\phi$ is a general CNF formula, then the implicational
dual rail encoding corresponding to $\phi$ is a Horn formula whose variables
correspond to the literals on $\var{\phi}$ and whose clauses
represent unit propagation on these literals. In addition, we
include the consistency clauses which are negative
clauses enforcing that at most one literal in each pair of complementary
literals on the variables $\var{\phi}$ is satisfied.

More precisely, assume, $\phi$ is a CNF formula on the variables $\vek{x}$
that does not contain an empty clause.
Let us introduce for
every $l \in \lit{\vek{x}}$ a meta-variable $\litVar{l}$ and let us
denote \(\meta(\vek{x})\) the set of all meta-variables.
For each pair $(C,l)$, such that $l \in C \in \phi$, we
express the step of unit propagation deriving $l$
from $\neg (C \setminus \{l\})$ as an implication
\begin{equation} \label{eq:implication}
   \left(\bigwedge_{e \in C \setminus \{l\}} \litVar{\neg e}\right)
   \to\litVar{l}
\end{equation}
or an equivalent Horn clause
\begin{equation} \label{eq:DR-clause}
   \litVar{l} \vee \bigvee_{e \in C \setminus \{l\}} \neg \litVar{\neg e} \,.
\end{equation}
The \emph{implicational dual rail encoding} of $\phi$
is the following formula on the variables \(\meta(\vek{x})\) consisting
of all the clauses~\eqref{eq:DR-clause} and the consistency clauses, namely
$$
\DR(\phi) = \bigwedge_{l \in C \in \phi}
\left( \litVar{l} \vee \bigvee_{e \in C \setminus \{l\}} \neg \litVar{\neg e} \right)
\wedge
\bigwedge_{x \in \vek{x}} (\neg \litVar{x} \vee \neg \litVar{\neg x}) \,.
$$

\subsection{q-Horn Formulas} %{
\label{sub:q-horn:def} % chktex 24

The class of q-Horn formulas was introduced in~\cite{BCH90} as a
generalization of both renamable Horn formulas and 2-CNF formulas.
It is a tractable class which
means that satisfiability of q-Horn formulas can be tested in polynomial
time, this class of formulas is closed under partial assignment, and we
can check if a formula is q-Horn in linear time by results of~\cite{BHS94}.
Although Horn formulas are URC, q-Horn formulas are not URC in general.
For example an unsatisfiable 2-CNF consisting only of binary clauses
is q-Horn, but it is not URC\@.

Following the
characterization of q-Horn formulas described in~\cite{BHS94}, let
us define these formulas as follows.

\begin{definition}[\cite{BCH90,BHS94}] \mlabel{def:q-horn}
   Let \(\phi\) be a CNF formula. We say that
   \(\gamma:\lit{\phi}\to\{0, \frac{1}{2}, 1\}\) is a valuation of
   literals in \(\phi\) if
   \(\gamma(u)+\gamma(\neg u)=1\) for every literal
   \(u\in\lit{\phi}\). Formula \(\phi\) is \emph{q-Horn} if
   there is a valuation \(\gamma\) of literals in \(\phi\) which
   satisfies for every clause \(C\in \phi\) that
   \begin{equation} \label{eq:q-horn}
      \sum_{u\in C}\gamma(u)\leq 1\,.
   \end{equation}
   The value \(\gamma(u)\) is called \emph{weight} of literal \(u\).
   If \(x\) is a variable, then \(\gamma(x)\) is called the
   \emph{weight} of variable \(x\).
\end{definition}

It is easy to observe that any renamable Horn formula is q-Horn, only
values \(0\) and \(1\) are sufficient in a valuation. Also any 2-CNF
is q-Horn, just assign \(\frac{1}{2}\) to every literal.

% vim: set filetype=tex :
% vim: set spell spelllang=en :

\section{Results}
\label{sec:results} % chktex 24

In this section, we state the results proven in the paper.

\subsection{Separations} \label{subsec:separations}

Let us first recall the notion of succinctness introduced
in~\cite{GKPS95} and used later extensively in~\cite{DM02}.

\begin{definition}[Succinctness]
   Let \(\mathbf{L}_1\) and \(\mathbf{L}_1\) be two representation
   languages. We say that \(\mathbf{L}_1\) is \emph{at least as succinct as}
   \(\mathbf{L}_2\), iff there exists a polynomial \(p\) such that for
   every sentence \(\varphi\in\mathbf{L}_2\), there exists an
   equivalent sentence \(\psi\in\mathbf{L}_1\) where \(|\psi|\leq
      p(|\varphi|)\).
   We  say that \(\mathbf{L}_1\) is \emph{strictly more succinct} than
   \(\mathbf{L}_2\) if \(\mathbf{L}_1\) is at least as succinct as
   \(\mathbf{L}_2\) but \(\mathbf{L}_2\) is not at least as succinct
   as \(\mathbf{L}_1\).
\end{definition}

In our results, we consider four representation languages --- URC
formulas and encodings and PC formulas and encodings. It follows
directly from the definitions that URC encodings are at least as
succinct as URC formulas, PC encodings are at least as succinct as PC
formulas, URC formulas are at least as succinct as PC formulas, and
URC encodings are at least as succinct as PC encodings.

It is a well-known fact that there are functions with a polynomial size PC encoding that
require an exponential size CNF formula without auxiliary variables even if
no level of the propagation strength is required. A folklore example is the parity
of $n$ variables $x_1 \oplus \ldots \oplus x_n$ which requires CNF formula of
size $2^{n-1}$ and has an encoding of linear size using auxiliary variables $y_i$
defined by the recurrence $y_0 = 0$
and $y_i = y_{i-1} \oplus x_i$. Moreover, the straightforward way to encode this
recurrence into a CNF formula together with the unit clause $y_n$ leads to
a PC encoding, see~\cite{BM12}.
As a consequence, URC encodings are strictly more succinct than URC formulas
and PC encodings are strictly more succinct than PC formulas.
We prove the following result on succinctness of URC and PC formulas without
auxiliary variables.

\begin{theorem} \mlabel{thm:sep:urc:pc}
The language of URC formulas is strictly more succinct than the language of PC
formulas.
\end{theorem}

Every PC formula is a URC formula. Consequently, Theorem~\ref{thm:sep:urc:pc}
follows from~\ref{enum:lb-horn:a} in the next proposition proven in Section~\ref{sec:lb-PC}.

\begin{theorem} \mlabel{thm:lb-horn}
For every $n \ge 1$, there is a Horn and, hence, a URC formula $\phi_n$
of size $O(n)$, such that
\begin{enumerate}[label={(\alph*)}]
\item\label{enum:lb-horn:a} every PC formula equivalent to $\phi_n$ has size $2^{\Omega(n)}$,
\item\label{enum:lb-horn:b} there is a PC encoding of $\phi_n$ of size $O(n)$.
\end{enumerate}
\end{theorem}

Theorem~\ref{thm:lb-horn} compares the succinctness of URC and PC formulas
and encodings of functions that can be represented by a polynomial size CNF formula.
Due to this, the theorem provides a separation between PC formulas
and PC encodings suitable when we consider compilation of a CNF formula of a moderate
size into a possibly larger formula or an encoding which is propagation complete.
A similar separation in case of URC formulas and URC encodings follows from the next
proposition proven in Section~\ref{sec:lb-URC}.

\begin{theorem} \mlabel{thm:lb-q-horn}
For every $n \ge 1$, there is a q-Horn formula $\phi_n$ of
size $O(n)$, such that
\begin{enumerate}[label={(\alph*)}]
\item every URC formula equivalent to $\phi_n$ has size $2^{\Omega(n)}$,
\item there is a PC and, hence, a URC encoding of $\phi_n$ of size $O(n)$.
\end{enumerate}
\end{theorem}

The examples of the formulas used for the separations
in theorems~\ref{thm:lb-q-horn} and~\ref{thm:lb-horn}
belong to well-known tractable classes. The results demonstrate that besides
the fact that Horn formulas are URC, no other propagation properties of
the formulas in these classes are guaranteed in the following strong sense:
an exponential number of additional implicates may be necessary to obtain
a PC formula equivalent to a given Horn formula or to obtain a URC formula
equivalent to a given q-Horn formula.

\subsection{Irredundant PC and URC Formulas}
\label{sub:irred-PC}

We prove a remarkable difference in the properties of the size of irredundant PC
and irredundant URC formulas representing a given function. A PC formula is
called a \emph{PC-irredundant} formula, if removing any clause either changes
the represented function or yields a formula that is not a PC formula
(such formulas were also called minimal propagation complete formulas
in~\cite{BM12}).
A \emph{URC-irredundant} formula is defined in a similar way. Namely, we prove:
\begin{itemize}
\item The sizes of any two PC-irredundant formulas for the same
      function differ at most by a factor polynomial in the number
      of the variables.
\item There are URC-irredundant formulas for the same function such that
      one has size polynomial and the other has size exponential
      in the number of the variables.
\end{itemize}
A PC-irredundant formula can be obtained from any PC formula in polynomial
time by repeated removal of \emph{absorbed clauses}, see~\cite{BM12}.
The first statement means that this is also a guarantee of a small
size of a PC formula, if we disregard factors polynomial in the number
of the variables. On the other hand, the second statement means that
irredundancy cannot be used in the same way for URC formulas.

The statement above concerning PC formulas follows from the next theorem
proven in Section~\ref{sub:proof-irred-PC}.

\begin{theorem} \mlabel{thm:PC-irredundant}
If $\phi_1$ and $\phi_2$ are PC-irredundant representations
of the same function of $n$ variables, then
$|\phi_2| \le n^2 |\phi_1|$.
\end{theorem}

The proof of this statement is similar to the proof of the known
result~\cite{HK93} that the sizes of two equivalent irredundant Horn formulas
differ at most by a factor at most $n-1$. For the PC formulas, the factor
is larger, since we have to work with absorbed clauses instead
of 1-provable implicates. This difference can be eliminated by replacing
the formula with its implicational dual rail encoding which represents a specific
Horn function. Using the characterization of PC formulas in
Theorem~\ref{thm:characterize-PC} below one can obtain a weaker
form of Theorem~\ref{thm:PC-irredundant} as a direct translation of
the properties of irredundant Horn formulas.

The statement above concerning URC formulas follows from the next proposition
proven in Section~\ref{sec:lb-irred-URC}.

\begin{proposition} \mlabel{prop:URC-irredundant}
For every $n \ge 1$, there is a PC formula $\phi_{n,1}$ of size $O(n)$ and
an equivalent URC-irredundant formula $\phi_{n,2}$ of size $2^{\Omega(n)}$.
\end{proposition}

\subsection{A Characterization of PC Formulas}
\label{sub:characterize-PC}

We use a simulation of the unit propagation in a CNF formula $\phi$
by an implicational encoding $\DR(\phi)$
described in Section~\ref{sub:dual-rail:def} and prove that
$\phi$ is a PC formula, if and only if $\DR(\phi)$
expresses a Horn function representing semantic consequence on the
literals implied by $\phi$. Since this Horn function is represented
by the implicational dual rail encoding associated
with the set of all the prime implicates of $\phi$, we
have the following theorem proven in Section~\ref{sec:characterize-PC}.

\begin{theorem} \label{thm:characterize-PC}
Let $\phi$ be a satisfiable CNF formula and let $\psi$ be the CNF formula
consisting of all prime implicates of $\phi$. Then $\phi$ is a PC formula
if and only if $\DR(\phi)$ and $\DR(\psi)$ are equivalent.
\end{theorem}

\subsection{An Encoding for q-Horn Formulas}
\label{sub:q-horn:result}

We prove that although q-Horn formulas can be strongly non-URC in the sense
formulated in Theorem~\ref{thm:lb-q-horn}, every q-Horn formula has a URC encoding
of size polynomial in the size of the original formula.

\begin{theorem}
   \mlabel{thm:qh-enc:main} % chktex 24
Assume that \(\varphi(\vek{x})\) is a q-Horn formula with a valuation $\gamma$
representing a q-Horn function \(f(\vek{x})\). Moreover, assume
a partition of the variables $\vek{x}=\vek{x}_1 \cup \vek{x}_2$, such that
$\gamma(x) \in \{0,1\}$ for every $x \in \vek{x}_1$ and
$\gamma(x) = 1/2$ for every $x \in \vek{x}_2$.
Then there is a URC encoding
   \(\psi(\vek{x}, \vek{y})\) of \(f(\vek{x})\) satisfying
   \(|\vek{y}|=O(|\vek{x}_2|^2)\),
   \(|\psi|=O(|\varphi|+|\vek{x}_2|^3)\), and
   \(\cnflen{\psi}=O(\cnflen{\varphi}+|\vek{x}_2|^3)\).
\end{theorem}

The encoding guaranteed in this theorem is not q-Horn in general.
Note also that the propagation strength obtained for this encoding is weaker than
what we obtain for the function used in Theorem~\ref{thm:lb-q-horn}, which has
a polynomial size PC encoding.

% vim: set filetype=tex :
% vim: set spell spelllang=en :

\section{Lower Bounds on the Size of Specific Formulas}
\label{sec:lb}

In this section, we present three examples of functions suitable
for proving a lower bound on the size of a specific formula. The
examples serve as proofs of theorems~\ref{thm:lb-horn}
and~\ref{thm:lb-q-horn}.

\subsection{Lower Bound on the Size of a PC Formula}
\label{sec:lb-PC}

In this section we prove Theorem~\ref{thm:lb-horn}.
Assume that $\phi$ is a satisfiable CNF formula not containing
the variable $x$ and let
\begin{equation} \label{eq:non-PC-1}
   \psi = \bigwedge_{C \in \phi} (\neg x \vee C) \,.
\end{equation}
One can verify that the set of prime implicates of $\psi$ is the set of all clauses
of the form $\neg x \vee C$, where $C$ is a prime implicate of $\phi$,
see also~\cite{BM12}.
Let us verify that the only prime PC representation of $\psi$ is
the set of all its prime implicates. Assume that $\psi'$ is a prime
formula equivalent to $\psi$. If a prime implicate $\neg x \vee D$
of $\psi$ is not in $\psi'$, let $\alpha$ be the
partial assignment $\neg D$. Clearly, $\psi' \wedge \alpha$
implies $\neg x$, however, every clause of $\psi'$ has the form
$\neg x \vee C$, where $C$ is a prime implicate of $\phi$ different
from $D$. It follows that $C$ contains at least one literal
not falsified by $\alpha$ and, hence, the unit propagation
does not derive any additional literal from $\alpha$ together
with $\neg x \vee C$. This implies that $\phi'$ is not a PC formula,
hence the only prime PC formula equivalent to $\psi$ consists of
all its prime implicates.

Each clause of the formula~\eqref{eq:non-PC-1} contains the literal $\neg x$.
Let us demonstrate that a lower bound on the size of a PC formula
based on the set of all prime implicates can be obtained also for
a formula, where different clauses are extended with literals on
different variables.

\begin{lemma} \mlabel{lem:lb-cycle}
Assume that $\phi = \bigwedge_{i=1}^m C_m$ is a satisfiable CNF
formula with $p$ prime implicates. Let $x_1, \ldots, x_m$ be new
variables and let
$$
\psi = \bigwedge_{i=1}^{m-1} (\neg x_i \vee x_{i+1}) \wedge
(\neg x_m \vee x_1) \wedge \bigwedge_{i=1}^m (\neg x_i \vee C_i)
$$
Then, the number of the prime implicates of $\psi$ is $mp + m(m-1)$
and the size of a smallest PC representation of $\psi$ is $p + m$.
\end{lemma}

\begin{proof}
Let $\Pi$ be the set of all prime implicates of $\phi$.
There are two types of prime implicates of $\psi$. The first type are
the clauses $\neg x_i \vee x_j$ where $i \not= j$, $1 \le i,j \le m$.
The second type are clauses $\neg x_i \vee C$ where $1 \le i \le m$
and $C \in \Pi$. Hence, the number of the prime implicates of $\psi$
is $mp + m(m-1)$.

Consider the formula
$$
\psi' = \bigwedge_{i=1}^{m-1} (\neg x_i \vee x_{i+1}) \wedge
(\neg x_m \vee x_1) \wedge \bigwedge_{C \in \Pi} (\neg x_1 \vee C) \,.
$$
This is a subset of the prime implicates of $\psi$ and all prime implicates
of $\psi$ can be obtained by non-merge resolution from it. It follows by results of~\cite{BM12} that
$\psi'$ is a PC representation of $\psi$. Since $\psi'$ has size $p + m$,
there is a PC representation of this size.

Assume, $\psi''$ is any prime PC representation of $\psi$. By applying
a satisfying assignment of $\phi$ to $\psi''$, the clauses of the first type
are unchanged and the clauses of the second type are satisfied. Since the
restricted formula represents the equivalence of the variables $x_i$ for
$1 \le i \le m$, the formula $\psi''$ contains at least $m$ clauses
of the first type.

Assume for a contradiction that there is an implicate $D \in \Pi$, such
that $\psi''$ does not contain any of the clauses $\neg x_i \vee D$.
Let $\alpha$ be the partial assignment $\neg D$. Clearly, $\psi \wedge \alpha$
implies $\neg x_1$. On the other hand, for each clause $\neg x_i \vee C$ in $\psi''$
at least one literal in $C$ is not falsified by $\alpha$, since $C$ and
$D$ are different prime implicates of $\phi$. Since the clause also contains
the literal $\neg x_i$, unit propagation does not derive any additional
literal from $\neg x_i \vee C$ and $\alpha$. This implies that $\alpha$ is
closed under unit propagation in the formula $\psi''$ and does not contain
$\neg x_1$. This is a contradiction. It follows that $\phi''$ contains at
least $p$ clauses of the second type and, hence, has size at least $p+m$.
\end{proof}

For every $m \ge 3$, let $\psi_m$ be the Horn formula of $3m+1$ variables
consisting of the clauses
$$
\begin{array}{ll}
  \neg x_i \vee \neg y_i \vee z_i & i = 1,\dots,m-1\\
  \neg x_m \vee \neg z_1 \vee \dots \vee \neg z_{m-1}\\
  \neg x_i \vee x_{i+1} & i = 1,\dots,m-1\\
  \neg x_m \vee x_1
\end{array}
$$

\begin{proposition} \label{prop:lb-PC}
Formula $\psi_m$ has $m2^{m-1} + O(m^2)$ prime implicates and its smallest
PC representation has size $2^{m-1} + O(m)$.
\end{proposition}

\begin{proof}
The formula $\psi$ has the form from Lemma~\ref{lem:lb-cycle}, if $\phi$ is
$$
\begin{array}{ll}
  \neg y_i \vee z_i & i = 1,\dots,m-1\\
  \neg z_1 \vee \dots \vee \neg z_{m-1}
\end{array}
$$
Since $m \ge 3$, formula $\phi$ has $2^{m-1} + m - 1$ prime implicates. Hence,
by the previous lemma, the number of the prime implicates of $\psi$
is $m(2^{m-1} + (m - 1)) + m(m - 1)$ and the smallest size of a PC formula
equivalent to $\psi$ is $2^{m-1} + 2m - 1$.
\end{proof}

\begin{lemma} \mlabel{lem:enc-horn}
There is a PC encoding of the function represented by $\psi_m$ of
size linear in the number of the variables.
\end{lemma}

\begin{proof}
Consider the order of the variables given as
$x_1, y_1, z_1, \ldots, x_{m-1}, y_{m-1}, z_{m-1}, x_m$.
One can verify that the standard construction of a Decision-DNNF using this
order of the variables in all paths leads to an OBDD of constant width.
This OBDD satisfies the assumption for the construction of the PC encoding
CompletePath described in~\cite{AGMES16}.
Since the OBDD has constant width, the encoding has size linear in the number of the
variables even if the exactly-one constraint used in CompletePath is represented by
a formula of quadratic size without auxiliary variables.
\end{proof}

\begin{proof}[Proof of Theorem~\ref{thm:lb-horn}]
Let $m=\max(n,3)$. By construction, $\psi_m$ is a Horn formula.
The condition~\ref{enum:lb-horn:a} follows from Proposition~\ref{prop:lb-PC}
and condition~\ref{enum:lb-horn:b} follows from Lemma~\ref{lem:enc-horn}.
\end{proof}

\subsection{Lower Bound on the Size of a URC Formula} %{
\label{sec:lb-URC}

In this section we shall proof Theorem~\ref{thm:lb-q-horn}.
Given a natural number \(n\), define a formula \(\psi_n(\vek{x}, \vek{a}, \vek{b})\)
where \(\vek{x}=(x_1, \dots, x_n)\), \(\vek{a}=(a_1, \dots, a_n)\), and
\(\vek{b}=(b_1, \dots, b_n)\) as follows
\begin{align*}
   \psi_n=&\bigwedge_{i=1}^{n-1}(\neg a_i\vee \neg x_i\vee x_{i+1})(\neg a_i\vee x_i\vee \neg x_{i+1})%
(\neg b_i\vee \neg x_i\vee x_{i+1})(\neg b_i\vee x_i\vee \neg x_{i+1})\\
   &\land(\neg a_n\vee \neg x_1\vee \neg x_n)(\neg a_n\vee x_1\vee x_n)
   (\neg b_n\vee \neg x_1\vee \neg x_n)(\neg b_n\vee x_1\vee x_n)\text{.}
\end{align*}
We can also write \(\psi_n\) more concisely as
\begin{equation} \label{eq:psi-concise}
   \psi_n\equiv\bigwedge_{i=1}^{n-1}\left[(a_i\lor b_i)\Rightarrow (x_i\Leftrightarrow x_{i+1})\right]
   \land\left[(a_n\lor b_n)\Rightarrow (x_1\Leftrightarrow \neg x_n)\right]\,.
\end{equation}
Observe that \(\psi_n\) is not URC, because
\(\psi_n\land\bigwedge_{i=1}^n a_i\models\bot\), however
\(\psi_n\land\bigwedge_{i=1}^n a_i\not\vdash_1\bot\). On the
other hand, one can verify that \(\psi_n\) is q-Horn using the valuation
\(\gamma(x_i)=\gamma(\neg x_i)=\frac{1}{2}\),
\(\gamma(a_i)=\gamma(b_i)=0\), and \(\gamma(\neg a_i)=\gamma(\neg
   b_i)=1\) for all \(i=1, \dots, n\).
Main part of Theorem~\ref{thm:lb-q-horn} is the following lower bound.

\begin{proposition}
   \mlabel{prop:q-horn-urc-lb}
   If \(\varphi\) is a URC formula equivalent to \(\psi_n\),
   then \(|\varphi|\geq 2^n\).
\end{proposition}
\begin{proof}
Without loss of generality, we can assume that $\phi$ is a prime formula.
Since all the ocurrences of the variables $\vek{x}$ in~\eqref{eq:psi-concise}
are involved in equivalences and non-equivalences, the set of satisfying
assignments of $\psi_n$ is invariant under taking the negation of all
the variables $\vek{x}$ simultaneously. This implies the following. Assume,
a prime implicate $C$ of $\phi_n$ contains precisely one literal $l \in \lit{\vek{x}}$.
Since $C$ is prime, there is a model of the formula, such that $l$ is the only
satisfied literal in $C$. This is not possible, since negating all $\vek{x}$ variables in the model
would lead to an assignment that is not a model of $\phi_n$. It follows that every
prime implicate of $\psi_n$ containing a literal from $\lit{\vek{x}}$ contains
at least two such literals.

Let \(U\) denote the set of partial assignments \(\alpha\cup \beta\)
where \(\alpha\subseteq \vek{a}\), \(\beta\subseteq \vek{b}\) and for
every index \(i=1, \dots, n\) exactly one of the variables
\(a_i\) and \(b_i\) belongs to \(\alpha\cup \beta\). In particular
\(|\alpha\cup \beta|=n\) and \(|U|=2^n\). Clearly, for every
$\alpha \cup \beta \in U$, the formula $\psi_n \wedge \alpha \wedge \beta$
is inconsistent. It follows that the clauses in the set
$\overline{U} = \{\neg(\alpha \wedge \beta) \mid \alpha\cup \beta \in U\}$
are implicates of $\psi_n$. Moreover, one can verify that $\overline{U}$
is precisely the set of all prime implicates
of $\psi_n$ containing only the literals from $\lit{\vek{a} \cup \vek{b}}$.
Let us prove that $\phi$ contains all of the implicates in $\overline{U}$.

Assume a partial assignment $\alpha \cup \beta \in U$ and the clause
$C = \neg(\alpha \wedge \beta)$. Assume for a contradiction that $C$
is not in $\phi$. Since $\phi \wedge \alpha \wedge \beta$ is inconsistent
and $\phi$ is URC, we have $\phi \wedge \alpha \wedge \beta \vdash_1 \bot$.
Using the argument from the first paragraph of the proof, this derivation does
not use any of the prime implicates containing a literal from $\lit{\vek{x}}$.
However, unit propagation using implicates from $\overline{U} \setminus \{C\}$
derives only negative literals on the variables $\vek{a} \cup \vek{b}$ and
these literals cannot be used to continue unit propagation using clauses
from $\overline{U} \setminus \{C\}$. This contradicts the assumption
$\phi \wedge \alpha \wedge \beta \vdash_1 \bot$. It follows that $C$ is
in $\phi$ as required.
\end{proof}

\begin{lemma} \mlabel{lem:enc-q-horn}
There is a PC encoding of the function represented by $\psi_n$ of
size linear in the number of the variables.
\end{lemma}

\begin{proof}
Use the same approach as in the proof of Lemma~\ref{lem:enc-horn} with the
order of the variables given as $x_1, a_1, b_1, x_2, a_2, b_2, \dots, x_n, a_n, b_n$.
\end{proof}

\begin{proof}[Proof of Theorem~\ref{thm:lb-q-horn}]
Consequence of Proposition~\ref{prop:q-horn-urc-lb} and
Lemma~\ref{lem:enc-q-horn}.
\end{proof}

The size of the encoding guaranteed by Lemma~\ref{lem:enc-q-horn} is
$Cn + O(1)$ with a relatively large constant $C$.
Let us note that a PC encoding of the function represented by
\(\psi_n\) of size $4n + O(1)$ can be described as follows. We use
additional auxiliary variables \(\vek{c}=(c_1, \dots, c_n)\) and
define
\begin{align*}
   \psi_n'=
   & \bigwedge_{i=1}^n (\neg a_i \lor c_i)(\neg b_i \lor c_i)\\
   & \land(\neg c_1\lor\dots\lor \neg c_n) \\
   & \land\bigwedge_{i=1}^{n-1}(\neg c_i\vee \neg x_i\vee x_{i+1})(\neg c_i\vee x_i\vee \neg x_{i+1})\\%
   &\land(\neg c_n\vee \neg x_1\vee \neg x_n)(\neg c_n\vee x_1\vee x_n)\text{.}
\end{align*}
One can check \(\psi_n'(\vek{x}, \vek{a}, \vek{b}, \vek{c})\) is a PC
encoding of the function represented by \(\psi_n(\vek{x}, \vek{a},
   \vek{b})\). We omit the proof because it is rather technical and
Lemma~\ref{lem:enc-q-horn} is sufficient for
the proof of Theorem~\ref{thm:lb-q-horn} above.

%} subsection q-horn:def

\section{Irredundant PC and URC Formulas}
\label{sec:irred-PC-URC}

In this section we prove the results formulated in Section~\ref{sub:irred-PC}.

\subsection{Irredundant PC Formulas}
\label{sub:proof-irred-PC}

\begin{proof}[Proof of Theorem~\ref{thm:PC-irredundant}]
If $C \in \phi_1$ and $l \in C$, then
$\phi_2 \wedge \neg (C \setminus \{l\}) \vdash_1 l$ or
$\phi_2 \wedge \neg (C \setminus \{l\}) \vdash_1 \bot$,
since $\phi_2$ is a PC formula
and $C$ is its implicate. Let $\theta_{C,l}$ be a set of clauses of $\phi_2$
used in some of these derivations. Since the derivation is unit resolution
and each clause is used to derive a literal on a different variable,
we have $|\theta_{C,l}| \le n$. Moreover, we have $\theta_{C,l} \models C$.
Let $\phi_2'$ be the union of $\theta_{C,l}$
for all $C \in \phi_1$ and $l \in C$. Clearly, $\phi_2'$ is a subset of
$\phi_2$ and $|\phi_2'| \le n^2 |\phi_1|$.
It remains to verify that $\phi_2'$ is a PC formula equivalent to $\phi_2$.
This implies that $\phi_2'= \phi_2$, since $\phi_2$ is PC-irredundant,
and the statement follows.

Formula $\phi_2'$ is implied by $\phi_2$ and implies $\phi_1$. Together,
this implies that $\phi_2'$ is equivalent to
both $\phi_1$ and $\phi_2$. Consider the formula $\phi_2' \wedge \phi_1$.
Since it contains $\phi_1$, it is a PC formula. By construction of $\phi_2'$,
each clause of $\phi_1$ is absorbed by $\phi_2'$, see \cite{BM12} for more
detail. This implies that we can successively remove all clauses of $\phi_1$
from $\phi_2' \wedge \phi_1$ while keeping its PC property.
\end{proof}

\subsection{A Large Irredundant URC Formula}
\label{sec:lb-irred-URC}

We call a formula URC-irredundant, if removing any clause
either changes the represented function or leads to
a formula that is not URC. In this section, we present
an example of a formula that can be extended by additional
clauses to a PC formula
of size linear in the number of the variables and has
also a URC-irredundant extension of size exponential
in the number of the variables.

Let $m \ge 2$. Consider the variables $a_i$, $b_i$, $c_i$, $d_i$,
the definite Horn formulas
$$
\delta_i = (\neg a_i \vee b_i) \wedge (\neg a_i \vee c_i)
\wedge (\neg b_i \vee \neg c_i \vee d_i)
$$
for $i=1,\ldots,m$, and the formulas
\begin{align*}
\gamma_m &= \left(\bigvee_{i=1}^m a_i\right) \wedge
\bigwedge_{i=1}^m \delta_i\\
\gamma'_m &= \gamma_m \wedge \bigwedge_{i=1}^m (\neg a_i \vee d_i)\\
\gamma''_m &= \gamma_m \wedge \bigwedge_{I \in E_m}
\left(\bigvee_{i \not\in I} a_i \vee \bigvee_{i \in I} d_i\right)
\end{align*}
where $E_m$ is the system of all non-empty subsets
$I \subseteq \{1,\ldots,m\}$ such that $|I|$ is even.
Clearly, we have
\begin{eqnarray*}
|\gamma_m|   & = & 3m + 1\\
|\gamma'_m|  & = & 4m + 1\\
|\gamma''_m| & = & 3m + 2^{m-1}
\end{eqnarray*}

\begin{lemma} \mlabel{lem:gamma-primes-equivalent}
$\gamma'_m$ and $\gamma''_m$ are equivalent to $\gamma_m$.
\end{lemma}

\begin{proof}
For every $i=1,\ldots,m$, the clause $(\neg a_i \vee d_i)$
can be obtained by resolution from the clause
$(\neg b_i \vee \neg c_i \vee d_i)$ and the clauses
$(\neg a_i \vee b_i)$, $(\neg a_i \vee c_i)$. This implies
that $\gamma'_m$ is equivalent to $\gamma_m$.

For every $I \subseteq \{1,\ldots,m\}$, not only if $I \in E_m$,
the clause
$$
\bigvee_{i \not\in I} a_i \vee \bigvee_{i \in I} d_i
$$
can be obtained by resolution from the clause $\bigvee_{i=1}^m a_i$
and the clauses $\neg a_i \vee d_i$ for $i \in I$.
This implies that $\gamma''_m$ is equivalent to $\gamma_m$.
\end{proof}

\begin{lemma} \mlabel{lem:gamma-prime-PC}
$\gamma'_m$ is a PC formula.
\end{lemma}

\begin{proof}
The disjunction $\bigvee_{i=1}^m a_i$ is a PC formula and also $\delta_i'=\delta_i\land(\neg a_i\lor d_i)$
are PC formulas. The formula $\gamma'_m$ is obtained from $\bigvee_{i=1}^m a_i$
by combining it sequentially with $\delta_i'$. Since in each step, we combine
formulas which have only one variable in common, the formula obtained in
each step is PC.
\end{proof}

\begin{lemma} \mlabel{lem:gamma-dprime-URC}
$\gamma''_m$ is a URC formula.
\end{lemma}

\begin{proof}
Assume a partial assignment $\rho$, such that
$\gamma''_m \wedge \rho$ is unsatisfiable and, hence,
also $\gamma_m \wedge \rho$ is unsatisfiable.
If there is an index $i \in \{1,\ldots,m\}$, such
that $\delta_i \wedge \rho$ is unsatisfiable, then a contradiction
can be derived by unit propagation, since $\delta_i$ is
a Horn and thus also URC formula. Assume, each of the formulas $\delta_i \wedge \rho$ is
satisfiable.
Assume for a contradiction that there is an index $j$, such that
$\rho \wedge \delta_j \not\models \neg a_j$.
The setting $a_j=1$ satisfies the
clause $\bigvee_{i=1}^m a_i$ in $\gamma_m$ and
we get a contradiction with the unsatisfiability of \(\gamma_m\land\rho\), since the formulas $\delta_i$ are
independent and each of them is satisfiable together with
$\rho \wedge a_j$. It follows that for all
$i=1,\ldots,m$, we have $\rho \wedge \delta_i \models \neg a_i$.
By case inspection, this is equivalent to the assumption that for
every $i=1,\ldots,m$, we have
$\{\neg a_i, \neg b_i, \neg c_i, \neg d_i\} \cap \rho \not= \emptyset$.
Let $J$ be the set of indices $i$, such that
$\delta_i \wedge \rho \not\vdash_1 \neg a_i$.
One can verify that for every $i \in J$, we have $\neg d_i \in \rho$.
This implies that the clause
$$
\bigvee_{i \not\in J} a_i \vee \bigvee_{i \in J} d_i
$$
is falsified by $\rho$ and the literals derived by unit propagation
limited to the subformulas $\delta_i$. If $J=\emptyset$, then this clause
is contained in $\gamma_m$. Otherwise, we distinguish two cases.
If $|J|$ is even, then $J \in E_m$ and the clause is one
of the additional clauses used to construct $\gamma''_m$.

If $|J|$ is odd, choose an arbitrary
$j \in J$ and consider $I=J \setminus \{j\}$.
The only literal in the clause
$$
\bigvee_{i \not\in I} a_i \vee \bigvee_{i \in I} d_i
\in \gamma''_m
$$
that is not falsified by unit propagation using $\rho$
and the formulas $\delta_i$ is $a_j$. Hence, unit propagation
using this clause derives $a_j$. Together with $\delta_j$,
we further derive $d_j$ and, finally, we
derive a contradiction, since $\neg d_j \in \rho$.
\end{proof}

\begin{lemma} \mlabel{lem:gamma-dprime-irred}
Every URC-irredundant subformula of $\gamma''_m$ has
size at least $2^{m-1}$.
\end{lemma}

\begin{proof}
Let $I \in E_m$ and let us prove that the formula
$\gamma''_m \setminus \{C\}$, where
$$
C = \bigvee_{i \not\in I} a_i \vee \bigvee_{i \in I} d_i
$$
is not a URC representation of $\gamma_m$.
Consider the partial assignment $\neg C$ and let
us prove that this assignment does not make any
of the clauses of $\gamma''_m \setminus \{C\}$
unit or empty. This can be verified by case inspection for
the clauses of $\delta_j$ for every $j=1,\ldots,m$.
Since $|I| \ge 2$, the clause
$$
\bigvee_{i=1}^m a_i \in \gamma_m
$$
also does not become unit or empty. If $J \in E_m \setminus \{I\}$,
then the number of the literals in the clause
$$
\bigvee_{i \not\in J} a_i \vee \bigvee_{i \in J} d_i
\in \gamma''_m
$$
that are not falsified by $\neg C$ is equal to the size
of the symmetric difference of the sets $I$ and $J$.
Since these sets are different and both have even size,
their symmetric difference is a non-zero even number, so
at least $2$. Altogether, the assignment $\neg C$ is closed
under unit propagation for the formula $\gamma''_m \setminus \{C\}$ and,
hence, unit propagation does not derive a contradiction from
$\neg C \wedge (\gamma''_m \setminus \{C\})$.
Since $C$ is an implicate of $\gamma_m$, this implies that
$\gamma''_m \setminus \{C\}$ is not a URC representation
of $\gamma_m$. This finishes the proof, since
the above argument can be used for $2^{m-1}$ clauses
$C \in \gamma''_m$.
\end{proof}

\begin{proof}[Proof of Proposition~\ref{prop:URC-irredundant}]
For a given $n$, let $m=\max(n, 2)$, $\phi_{n,1}=\gamma'_m$,
and let $\phi_{n,2}$ be any URC-irredundant subformula of $\gamma''_m$.
The required properties follow from
lemmas~\ref{lem:gamma-prime-PC}, \ref{lem:gamma-dprime-URC},
and~\ref{lem:gamma-dprime-irred}.
\end{proof}

Let us note that $\gamma''_m$ is a URC-irredundant formula.
None of the clauses of $\delta_i$ for $i \in \{1,\ldots,m\}$ can
be removed, since this changes the restriction of the function obtained
by setting all the variables $a_j, b_j, c_j, d_j$ for $j \not=i$ to $1$.
The clause $(\bigvee_{i=1}^m a_i)$ also cannot be removed, since
this is the only unsatisfied clause, if all $a_i$ variables are $0$
and all the variables $b_i, c_i, d_i$ are $1$.

\section{Characterization of PC Formulas}
\label{sec:characterize-PC}

We prove a characterization of PC formulas formulated in Section~\ref{sub:characterize-PC}
by relating their implicational dual rail encoding to a Horn function corresponding to the semantic
consequence on the literals implied by the formula.

Let \(\varphi(\vek{x})\) be a formula on a set of variables
\(\vek{x}\) and let \(\alpha\subseteq\lit{\vek{x}}\) be a set of
literals. The \emph{closure of \(\alpha\) under unit propagation in
   \(\varphi\)} is defined as
$$
\clup(\phi, \alpha) = \{l \in \lit{\vek{x}} \mid  \text{\(\varphi\land\alpha \vdash_1 l\) or
   \(\varphi\land\alpha\vdash_1 \bot\)}\} \,.
$$
If $\phi \wedge \alpha \vdash_1 \bot$, then $\clup(\phi, \alpha)$ contains
all the literals and not only those derived by unit propagation. This
is useful for the characterization of PC formulas. In contradictory cases,
different PC formulas may derive different sets of literals by unit propagation,
however, the closure as defined above is the same.
In non-contradictory cases, $\clup(\phi, \alpha)$ is precisely the set of
literals derived by unit propagation.

Assume, $f(\vek{x})$ is a boolean function on $\vek{x}$ and let
\(\alpha\subseteq\lit{\vek{x}}\) be a set of literals.
The \emph{semantic closure of
   $\alpha$ with respect to $f(\vek{x})$} is the set of literals defined as
$$
\clsem(f,\alpha) = \{l \in \lit{\vek{x}} \mid f(\vek{x})\land\alpha \models l\} \,.
$$
If $\phi$ represents $f$, we can also use
$\clsem(\phi, \alpha)$ instead of $\clsem(f, \alpha)$.
For every set of literals $\alpha$ and a formula $\phi$, we have
$\clup(\phi, \alpha) \subseteq \clsem(\phi, \alpha)$, since the literals
derived by unit resolution are also semantic consequences.

\begin{lemma} \mlabel{lem:equal-closures}
A formula $\phi$ is a PC formula, if
and only if for every partial assignment $\alpha$, we have
$\clup(\phi, \alpha) = \clsem(\phi, \alpha)$.
\end{lemma}

\begin{proof}
If $\phi$ is a PC formula, then
$\clup(\phi, \alpha) = \clsem(\phi, \alpha)$
can be proven by considering the cases
$\phi \wedge \alpha \models \bot$ and
$\phi \wedge \alpha \not\models \bot$ separately.
For the opposite direction, assume
$\clup(\phi, \alpha) = \clsem(\phi, \alpha)$.
Then for every $l$, such that $\phi \wedge \alpha \models l$,
we have $\phi \wedge \alpha \vdash_1 l$ or
$\phi \wedge \alpha \vdash_1 \bot$ as a direct consequence of
the definitions.
\end{proof}

One can verify that both $\clup$ and $\clsem$ are closure operators
on sets of literals. In particular, the system of sets
closed under any of them is closed under set intersection.
Moreover, it is well-known that two closure operators
on the subsets of the same set are equal, if and only if
they have the same closed sets. Since $\lit{\vek{x}}$
is the only closed set containing complementary literals
for both $\clup$ and $\clsem$, it is sufficient
to compare only closed sets that are partial assignments. In order
to represent partial assignments, we use the meta-variables
introduced in Section~\ref{sub:dual-rail:def}.
An assignment of these variables is the characteristic
function of a set of literals. Using this,
a partial assignment of the variables $\vek{x}$ is a total
assignment of the variables $\meta(\vek{x})$ satisfying
for every $x \in \vek{x}$ the consistency clause
$\neg \litVar{x} \vee \neg \litVar{\neg x}$.

\begin{definition} \mlabel{def:function-h}
   We denote $S(f)$ the set of partial assignments $\alpha$
   satisfying $\clsem(f, \alpha) = \alpha$ and $h_f$ the
   characteristic function of $S(f)$ represented as a boolean function
   on the variables $\meta(\vek{x})$.
\end{definition}

\begin{lemma} \mlabel{lem:closure-intersection}
   The function $h_f$ is a Horn function.
\end{lemma}

\begin{proof}
   The system of closed sets of $\clsem$ is $S(f) \cup \{\lit{\vek{x}}\}$
   and it is closed under set intersection.
   Since the intersection of any partial assignments $\alpha$, $\beta$
   is not equal to $\lit{\vek{x}}$, also $S(f)$ is closed under set intersection.
   The lemma follows, since a boolean
   function is a Horn function, if and only if the set of
   its satisfying assignments is closed under componentwise
   conjunction~\cite{CH11}. This corresponds to set intersection
   of the sets of literals represented by their characteristic functions.
\end{proof}

We use the implicational dual rail encoding
described in Section~\ref{sub:dual-rail:def} to capture
the connection of the propagation complete formulas for \(f\)
and the function \(h_f\).
The following property of the implicational encoding is crucial for
our purposes.

\begin{proposition} \mlabel{prop:characterize-DR}
   If $\phi(\vek{x})$ is a CNF formula not containing the empty clause, then
   the satisfying assignments of $\DR(\phi)$ are exactly the
   characteristic functions of partial assignments $\alpha$ of the
   variables in $\vek{x}$ satisfying $\clup(\phi, \alpha) = \alpha$.
\end{proposition}

\begin{proof}
   The consistency clauses in $\DR(\phi)$ are satisfied exactly
   by the characteristic functions of partial assignments.
Consider a partial assignment $\alpha$.
   The remaining clauses are equivalent to implications
   (\ref{eq:implication}).
   The conjunction of these implications expresses the condition
that every literal derivable by unit propagation from $\alpha$ belongs to $\alpha$.
Let us prove by contradiction that this implies
$\phi \wedge \alpha \not\vdash_1 \bot$.
If $\phi \wedge \alpha \vdash_1 \bot$, then there
is a non-empty clause $e_1 \vee \ldots \vee e_k \in \phi$ and the literals
$\neg e_1, \ldots, \neg e_k$ can be derived from $\phi \wedge \alpha$
by unit propagation. However, this implies that
$e_1$ is derived by unit propagation,
so $\alpha$ is not a partial assignment closed under the
rules~\eqref{eq:implication}.
It follows that $\phi \wedge \alpha \not\vdash_1 \bot$ and
$\clup(\phi, \alpha) = \alpha$, since in this case, $\clup$ is the
same as the closure under the rules~\eqref{eq:implication}.
The opposite direction is similar.
\end{proof}

Let \(\varphi(\vek{x})\) be a formula representing a function
\(f(\vek{x})\).  If a partial assignment $\alpha\subseteq\lit{\vek{x}}$
satisfies $\clsem(\phi, \alpha) = \alpha$, then $\phi \wedge \alpha$
is consistent and $\clup(\phi, \alpha) = \alpha$.
As a consequence, we have $\DR(\phi) \ge h_f$ in the sense that this
inequality holds for each assignment of the variables $\meta(\vek{x})$.
Moreover, we can characterize PC formulas representing a given function $f$
as follows.

\begin{proposition} \mlabel{prop:characterize-PC}
If $\phi$ is a CNF formula not containing the empty clause
and $f$ a boolean function, then $\phi$ is a PC formula
representing $f$, if and only if $\DR(\phi)$ represents $h_f$.
\end{proposition}

\begin{proof}
   Assume, $\DR(\phi)$ represents $h_f$.
   The substitutions $\litVar{x} \leftarrow x$ and $\litVar{\neg x} \leftarrow \neg x$
   for all $x \in \vek{x}$ transform $\DR(\phi)$ into $\phi$ and
   $h_f$ into $f$. It follows that $\phi$ represents $f$.
   By definition of $h_f$ and Proposition~\ref{prop:characterize-DR},
   the equivalence between $\DR(\phi)$ and $h_f$ implies
   that the closure operators $\clup$ and $\clsem$ for the
   formula $\phi$ have the same closed sets and, hence, are the same.
Together with Lemma~\ref{lem:equal-closures}, this implies
that $\phi$ is a PC formula.

   If $\phi$ is a PC representation of $f$, then one can prove that
   $\DR(\phi)$ and $h_f$ describe the same set of partial
   assignments using the same steps as above in the reversed order.
\end{proof}

\begin{proof}[Proof of Theorem~\ref{thm:characterize-PC}]
Assume that $\phi$ and $\psi$ satisfy the assumption and let $f$
be the function they represent.
Since $\psi$ is a PC representation of $f$,
Proposition~\ref{prop:characterize-PC} with $\phi$ replaced by $\psi$
implies that $\DR(\psi)$ represents $h_f$.
The proof of the theorem is then finished by Proposition~\ref{prop:characterize-PC}
used for the formula $\phi$.
\end{proof}

% vim: set filetype=tex :
% vim: set spell spelllang=en :

\section{URC Encoding of a q-Horn Formula} %{
\label{sec:q-horn:urc-enc} % chktex 24

In this section, we show that every q-Horn function represented by
a q-Horn formula \(\varphi(\vek{x})\) has a URC encoding of size polynomial
in the size of $\phi$. Let us fix a valuation \(\gamma\) which
shows that \(\varphi\) is q-Horn. For simplicity, we assume
\(\gamma(x)\geq\gamma(\neg x)\) for every variable \(x\in\vek{x}\). If
this assumption is not satisfied for a variable \(x\), we can replace
all occurrences of \(x\) with \(\neg x\) and vice versa and set
\(\gamma(x)=1-\gamma(x)\). As a consequence, all variables in $\phi$
have weight \(1\) or \(\frac{1}{2}\). The set of variables of weight
\(1\) and \(\frac{1}{2}\) will be denoted as \(\vek{x}_1\) and
\(\vek{x}_2\), respectively.

Given a CNF formula \(\varphi(\vek{x})\) and a partial assignment
\(\beta\subseteq\lit{\vek{x}}\) we denote \(\varphi(\beta)\) the
formula which originates from \(\varphi\) by applying partial
assignment \(\beta\). In particular, the clauses satisfied by
\(\beta\) are removed from \(\varphi\) and negations of literals in
\(\beta\) are removed from the remaining clauses. Let us recall that a
partial assignment \(\beta\subseteq\lit{\vek{x}}\) is an \emph{autark
   assignment}~\cite{K00} for a CNF \(\varphi\) if for every clause
\(C\in\varphi\) such that $C(\beta)$ is different from $C$, we have
that \(C\) is satisfied by \(\beta\). In particular, \(\varphi\) is
satisfiable if and only if \(\varphi(\beta)\) is satisfiable.

We split the formula \(\varphi(\vek{x}_1, \vek{x}_2)\) into two subformulas
$\phi_1(\vek{x}_1)$ and $\phi_2(\vek{x}_1, \vek{x}_2)$. A clause
\(C\in \varphi\) belongs to $\phi_1$ if \(\var{C}\subseteq\vek{x}_1\)
and it belongs to $\phi_2$ otherwise. By construction, \(\varphi_1\)
is a Horn formula. Clauses in \(\varphi_2\) contain
one or two variables of weight \(\frac{1}{2}\). Other literals in
these clauses have weight \(0\) and they are negations of variables
of weight \(1\). Let us recall Algorithm~\ref{alg:q-horn-sat} for checking
satisfiability of a q-Horn formula \(\varphi\) described
in~\cite{BCH90}. The input to the satisfiability checking procedure is
a q-Horn formula \(\varphi\) and a valuation \(\gamma\)
satisfying the assumption formulated above.
The proof of correctness of Algorithm~\ref{alg:q-horn-sat}
relies on the following property of
\(\varphi(\beta) = \phi_1(\beta) \wedge \phi_2(\beta)\). If \(C\) is
a clause in \(\varphi(\beta)\) which
contains a literal on a variable in \(\vek{x}_1\), then \(C\) contains a
negative literal on a (possibly different) variable in \(\vek{x}_1\).
It follows that these
clauses can be satisfied by an autark assignment that assigns $0$
to all the remaining variables $\vek{x}_1$.

\begin{algorithm}[t]
   \DontPrintSemicolon%
   \KwIn{q-Horn formula \(\varphi(\vek{x})\), valuation $\gamma$ satisfying
\(\gamma(x)\in\{\frac{1}{2}, 1\}\) for all \(x\in\vek{x}\)}
   \KwOut{\emph{SAT} if \(\varphi\) is satisfiable, \emph{UNSAT}
      otherwise}
   let $\phi_1$, $\phi_2$, and $\vek{x}_1$ be as described in the text\;
   \lIf{\(\varphi_1\vdash_1\bot\)}{\Return{UNSAT}}
   \(\beta \gets\{u \in \lit{\vek{x}_1} \mid\varphi_1 \vdash_1 u\}\)\;
   \(\varphi_2'\gets\{C\in\varphi_2(\beta)\mid C\subseteq\lit{\vek{x}_2}\}\)\label{step:q-horn-sat:4}\;
   \leIf{\(\varphi_2'\) is satisfiable\label{step:q-horn-sat:5}}{\Return{SAT}}{\Return{UNSAT}}
   \caption{Satisfiability checking of a q-Horn
      formula~\cite{BCH90}\label{alg:q-horn-sat}}
\end{algorithm}

\begin{table}[b]
   \begin{center}
      \begin{tabular}{Q l l}
         \toprule
         \multicolumn{1}{c}{group} & clause & condition\\
         \midrule
         \label{groupQH1} % chktex 24
         & \(C\) & \(C\in\varphi\) and \(|\var{C}\cap\vek{x}_2|\leq 1\)\\
         \label{groupQH2} % chktex 24
         & \(\neg x_{i_1}\lor\dots\lor \neg x_{i_k}\lor\litVar{u\lor v}\)
         & \(\neg x_{i_1}\lor\dots\lor \neg x_{i_k}\lor u\lor v\in\varphi_2\)\\
         \label{groupQH3} % chktex 24
         & \(\neg \litVar{u\lor v}\lor\neg\litVar{\neg v\lor
               w}\lor\litVar{u\lor w}\) & \(u\lor v, \neg v\lor w,
               u\lor w\in \varphi_q^+\)\\
         \label{groupQH4} % chktex 24
         & \(\neg\litVar{u\lor v}\lor\neg\litVar{u\lor
               \neg v}\lor u\) & \(u\lor v, u\lor \neg v\in \varphi_q^+\)\\
         \label{groupQH5} % chktex 24
         & \(\neg\litVar{u \lor v}\lor v\lor u\) &
         \(u \lor v\in \varphi_q^+\)\\
         \label{groupQH6} % chktex 24
         & \(\neg u\lor\litVar{u \lor v}\) & \(u\lor v\in\varphi_q^+\)\\
         \bottomrule
      \end{tabular}
   \end{center}
   \caption{The clauses of encoding \(\psi(\vek{x}, \vek{y})\) for a
      q-Horn formula \(\varphi(\vek{x})\), where $u$ and $v$ denote
arbitrary literals from $\lit{\vek{x}_2}$.}\label{tab:qh-enc:clauses}
\end{table}

We shall construct a URC encoding $\psi(\vek{x}, \vek{y})$ for a given
q-Horn formula \(\varphi(\vek{x})\). Unit propagation in the encoding
allows to simulate Algorithm~\ref{alg:q-horn-sat} with a formula
$\phi \wedge \alpha$ as an input, where $\alpha \subseteq \lit{\vek{x}}$.
Formula \(\varphi_2'\) created in Step~\ref{step:q-horn-sat:4}
is a 2-CNF formula on variables \(\vek{x}_2\) and
we can thus check its satisfiability in linear time in
Step~\ref{step:q-horn-sat:5}, see for example~\cite{APT79}. In the URC
encoding, Step~\ref{step:q-horn-sat:5} is implemented by simulating
a resolution derivation of a contradiction from \(\varphi_2'\) in
which only resolvents of size at most 2 are needed. The number of such
resolvents is at most quadratic in the number of the variables in \(\vek{x}_2\).
We can thus encode the resolution rules into a polynomial number of
clauses of the encoding.

Let us first introduce a necessary notation. Given a q-Horn formula
\(\varphi = \phi_1 \wedge \phi_2\),
we denote \(\varphi_q=\{C\cap\lit{\vek{x}_2}\mid C\in\varphi_2\}\)
which is a 2-CNF formula on variables of weight \(\frac{1}{2}\).
Moreover, let \(\varphi_q^+\) be the set of all clauses of size \(2\)
which can be derived by resolution from \(\varphi_q\) and we associate a
meta-variable \(\litVar{C}\) with every clause \(C\in \varphi_q^+\).
Note that if we test satisfiability of $\phi \wedge \alpha$
instead of $\phi$ in Algorithm~\ref{alg:q-horn-sat}, the formula $\phi_2'$
can contain additional unit clauses from $\alpha \cap \lit{\vek{x}_2}$,
however, all its quadratic clauses belong to $\phi_q$. The encoding
\(\psi(\vek{x}, \vek{y})\) for \(\varphi(\vek{x})\) uses the variables
\[\vek{y}=\{\litVar{C}\mid C\in\varphi_q^+\}\]
as auxiliary variables and consists of the clauses in Table~\ref{tab:qh-enc:clauses}.

Let us look more closely to clauses of the encoding.
Group~\ref{groupQH1} consists of all clauses of $\phi_1$ and some
of the clauses of $\phi_2$.
The clauses of~\ref{groupQH1} that belong to
\(\varphi_1\) allow unit
propagation on the clauses of \(\varphi_1\), thus implementing the
first part of Algorithm~\ref{alg:q-horn-sat}. Clauses of
group~\ref{groupQH1} that belong to \(\varphi_2\) and clauses of
group~\ref{groupQH2} allow to derive by unit propagation a representation
of the clauses of \(\varphi_2'\). Unit clauses are represented directly
and each binary clause in \(\varphi_2'\) is represented by the
corresponding meta-variable.
Clauses of groups~\ref{groupQH3} to~\ref{groupQH5}
allow to simulate resolution on binary clauses in \(\varphi_2'\)
represented by meta-variables by unit propagation in \(\psi(\vek{x}, \vek{y})\).
In particular, clause
\(\neg \litVar{u\lor v}\lor\neg\litVar{\neg v\lor w}\lor\litVar{u\lor w}\) in
group~\ref{groupQH3} allows to derive the literal \(\litVar{u\lor w}\)
representing a resolvent assuming the literals \(\litVar{u\lor v}\) and
\(\litVar{\neg v\lor w}\) representing the original clauses. Similarly,
clause \(\neg\litVar{u\lor v}\lor\neg\litVar{u\lor \neg v}\lor u\) in
group~\ref{groupQH4} allows to
derive resolvent \(u\) assuming \(\litVar{u\lor v}\) and
\(\litVar{u\lor\neg v}\). Finally, clause \(\neg\litVar{u \lor
      v}\lor v\lor u\) in group~\ref{groupQH5} allows to derive
resolvent \(u\) assuming \(\litVar{u\lor v}\) and \(\neg v\).
Clauses of groups~\ref{groupQH5} and~\ref{groupQH6} together represent
the equivalences \(u\lor v\Leftrightarrow \litVar{u\lor v}\) for all
clauses in $\varphi_q^+$. They thus
define the semantics of the meta-variables \(\litVar{u\lor v}\)
in $\vek{y}$.

\begin{lemma}
   \mlabel{lem:qh-enc:enc} % chktex 24
   Formula \(\psi(\vek{x}, \vek{y})\) is an encoding of
   \(\varphi(\vek{x})\).
\end{lemma}
\begin{proof}
The clauses of groups~\ref{groupQH5} and~\ref{groupQH6} imply
the clauses of groups~\ref{groupQH3} and~\ref{groupQH4}.
Moreover, they also imply that
the clauses of groups~\ref{groupQH1} and~\ref{groupQH2} are
equivalent to $\phi(\vek{x})$. It follows that $\psi(\vek{x}, \vek{y})$
is equivalent to the conjunction of $\phi(\vek{x})$ and the clauses
of groups~\ref{groupQH5} and~\ref{groupQH6}. This conjunction is
clearly an encoding of $\phi(\vek{x})$ obtained by adding definitions
of the new variables $\vek{y}$.
\end{proof}

\begin{lemma}
   \mlabel{lem:qh-enc:urc-enc} % chktex 24
Let \(\alpha\subseteq\lit{\vek{x}}\) be a partial assignment.
If \(\psi(\vek{x}, \vek{y})\land\alpha\models\bot\), then
   \(\psi(\vek{x}, \vek{y})\land\alpha\vdash_1\bot\).
\end{lemma}
\begin{proof}
Assume, \(\psi(\vek{x}, \vek{y})\land\alpha\models\bot\).
By Lemma~\ref{lem:qh-enc:enc}, we have \(\phi(\vek{x})\land\alpha\models\bot\).
If $\gamma$ satisfies (\ref{eq:q-horn}) for the formula $\phi$,
it satisfies (\ref{eq:q-horn}) also for the formula $\phi \wedge \alpha$.
It follows that Algorithm~\ref{alg:q-horn-sat} detects unsatisfiability of
\(\varphi\land\alpha\) using the valuation $\gamma$ derived originally
for $\phi$. Let $\alpha_1=\alpha \cap \lit{\vek{x}_1}$ and
$\alpha_2=\alpha \cap \lit{\vek{x}_2}$.
When used for $\phi \wedge \alpha$, Algorithm~\ref{alg:q-horn-sat}
uses $\phi_i \wedge \alpha_i$ instead of $\phi_i$ for $i=1,2$.

If \(\varphi_1\land\alpha_1\vdash_1\bot\), then also
\(\psi(\vek{x}, \vek{y})\land\alpha\vdash_1\bot\) since
   \(\varphi_1 \wedge \alpha_1 \subseteq \psi \wedge \alpha\) due to clauses in
   group~\ref{groupQH1}. Now assume
   \(\varphi_1\land\alpha_1\not\vdash_1\bot\).
Consider the assignment $\assign{b}_1$ used to obtain
$\varphi_2' \wedge \alpha_2 = (\phi_2 \wedge \alpha_2)(\assign{b}_1)$
in Algorithm~\ref{alg:q-horn-sat}.
Since $\assign{b}_1$ is an autark assignment for the formula
$\phi \wedge \alpha$ and satisfies all clauses of $\phi_1 \wedge \alpha_1$,
we have \(\varphi_2'\land\alpha_2\models\bot\). We claim that in this
   case for any clause \(u\lor v\in\varphi_2'\) we have
   \(\psi(\vek{x}, \vek{y})\land\alpha\vdash_1\litVar{u\lor v}\).
   Since \(u\lor v\in\varphi_2'\), there must be a clause \(C=\neg
      x_{i_1}\lor\dots\lor\neg x_{i_k}\lor u\lor v\) in \(\varphi_2\)
   and \(\varphi_1\land\alpha_1\vdash_1 x_{i_j}\) for every \(j=1,
      \dots, k\). Using clauses of group~\ref{groupQH1} we derive that
   also \(\psi(\vek{x}, \vek{y})\land\alpha\vdash_1 x_{i_j}\) for
   every \(j=1, \dots, k\) and then using a clause of
   group~\ref{groupQH2} corresponding to \(C\) we get \(\psi(\vek{x},
      \vek{y})\land\alpha\vdash_1 \litVar{u \lor v}\). Similarly,
   using clauses of group~\ref{groupQH1} in $\phi_2$, we obtain
   \(\psi(\vek{x}, \vek{y})\land\alpha\vdash_1 u\) for any unit clause
   \(u\in \varphi_2'\). Unit clauses in $\alpha_2$ are contained
in $\psi(\vek{x}, \vek{y})\land\alpha$ directly. Since
   \(\varphi_2'\land\alpha_2\models\bot\), there is a literal
   \(u\in\lit{\vek{x}_2}\), such that unit clauses \(u\) and
   \(\neg u\) can be derived by resolution from
   \(\varphi_2'\land\alpha_2\). Clauses of
   groups~\ref{groupQH3} to~\ref{groupQH5} allow to simulate
   resolution on at most binary clauses using unit propagation.
Using this, we obtain \(\psi(\vek{x},
      \vek{y})\land\alpha\vdash_1 u\) and \(\psi(\vek{x},
      \vek{y})\land\alpha\vdash_1 \neg u\). Together we obtain
   \(\psi(\vek{x}, \vek{y})\land\alpha\vdash_1 \bot\).
\end{proof}

Next, we show that \(\psi(\vek{x}, \vek{y})\) is a URC formula. We will
first show that we can allow positive occurrences of variables
from \(\vek{y}\) in the partial assignment.

\begin{lemma}
   \mlabel{lem:qh-enc:urc-pos-y} % chktex 24
   Let \(\alpha\subseteq\lit{\vek{x}}\cup\vek{y}\) be a partial
   assignment. If \(\psi(\vek{x}, \vek{y})\land\alpha\models\bot\),
   then
   \(\psi(\vek{x}, \vek{y})\land\alpha\vdash_1\bot\).
\end{lemma}
\begin{proof}
Assume \(\psi(\vek{x}, \vek{y})\land\alpha\models\bot\).
   Let us split \(\alpha\) in two partial assignments
   \(\alpha_x=\alpha\cap\lit{\vek{x}}\) and
      \(\alpha_y=\alpha\cap\vek{y}\).
The formula
   \begin{equation*}
      \varphi'(\vek{x})=\varphi(\vek{x})\land \bigwedge_{\litVar{C}\in\alpha_y}C
   \end{equation*}
is q-Horn using the valuation $\gamma$, since we add to $\phi$
binary clauses on the variables $\vek{x}_2$.
The encoding of \(\varphi'(\vek{x})\) constructed
according to Table~\ref{tab:qh-enc:clauses} is
   \(\psi'(\vek{x}, \vek{y})=\psi(\vek{x}, \vek{y})\land \alpha_y\).
   Since \(\alpha=\alpha_x\cup\alpha_y\), we have \(\psi(\vek{x},
      \vek{y})\land\alpha=\psi'(\vek{x}, \vek{y})\land\alpha_x\) and
   thus \(\psi'(\vek{x}, \vek{y})\land\alpha_x\models\bot\). By
   Lemma~\ref{lem:qh-enc:urc-enc} we get that \(\psi'(\vek{x},
      \vek{y})\land\alpha_x\models\bot\) implies
   \(\psi'(\vek{x}, \vek{y})\land\alpha_x\vdash_1\bot\). This is
   equivalent to
   \(\psi(\vek{x}, \vek{y})\land\alpha\vdash_1\bot\).
\end{proof}

\begin{lemma}
   \mlabel{lem:qh-enc:urc} % chktex 24
   Let \(\alpha\subseteq\lit{\vek{x} \cup \vek{y}}\) be a partial
   assignment. If \(\psi(\vek{x}, \vek{y})\land\alpha\models\bot\),
   then \(\psi(\vek{x}, \vek{y})\land\alpha\vdash_1\bot\).
\end{lemma}
\begin{proof}
Assume \(\psi(\vek{x}, \vek{y})\land\alpha\models\bot\).
   Let us split \(\alpha\) into three partial assignments
   \(\alpha_x=\alpha\cap\lit{\vek{x}}\),
   \(\alpha_y=\alpha\cap\vek{y}\), and
   \(\alpha_{\overline{y}}=\alpha\cap\{\neg y\mid y\in\vek{y}\}\).
Moreover, let
   \(\alpha_{\overline{y}}'=\{\neg u, \neg v\mid\neg \litVar{u\lor v}\in\alpha_{\overline{y}}\}\).
   Since $\psi(\vek{x}, \vek{y}) \models u\lor v\Leftrightarrow \litVar{u\lor v}$, we have
   $\psi(\vek{x}, \vek{y}) \models \alpha_{\overline{y}} \Leftrightarrow \alpha_{\overline{y}}'$.
   It follows that
   \(\psi(\vek{x}, \vek{y})\land\alpha_x\land\alpha_{y}\land\alpha_{\overline{y}}'\models\bot\).

   By Lemma~\ref{lem:qh-enc:urc-pos-y} we get
   \(\psi(\vek{x},
      \vek{y})\land\alpha_x\land\alpha_{y}\land\alpha_{\overline{y}}'\vdash_1\bot\).
   Using clauses of groups~\ref{groupQH5} and~\ref{groupQH6}, this is equivalent to
   \(\psi(\vek{x},
      \vek{y})\land\alpha_x\land\alpha_{y}\land\alpha_{\overline{y}}\vdash_1\bot\) and thus to
   \(\psi(\vek{x},
      \vek{y})\land\alpha\vdash_1\bot\) as required.
\end{proof}

We are now ready to prove the main result of this section formulated
in Section~\ref{sub:q-horn:result}.

\begin{proof}[Proof of Theorem~\ref{thm:qh-enc:main}]
The existence of the required encoding follows from
   lemmas~\ref{lem:qh-enc:enc} and~\ref{lem:qh-enc:urc}. The size estimates follow directly from
   the construction described in Table~\ref{tab:qh-enc:clauses}.
\end{proof}

Let us point out that the encoding in Table~\ref{tab:qh-enc:clauses}
is not q-Horn in general. If the groups~\ref{groupQH5} and~\ref{groupQH6}
of clauses are present in the encoding for some literals $u$, $v$, they
contain clauses
$\neg\litVar{u \lor v}\lor v \lor u$,
$\neg u\lor\litVar{u \lor v}$,
$\neg v\lor\litVar{u \lor v}$.
One can easily verify that a valuation $\gamma$ satisfying~\eqref{eq:q-horn}
for these clauses has to satisfy
$\gamma(u) = \gamma(v) = \gamma(\litVar{u \lor v})=0$.
The variable $\litVar{u \lor v}$ is included in the encoding,
if the original q-Horn formula requires $\gamma(u) = \gamma(v) = 1/2$.
In this case, the system of inequalitites~\eqref{eq:q-horn}
for the encoding in Table~\ref{tab:qh-enc:clauses} is inconsistent.

%} subsection q-horn:urc-enc

\section{Conclusion and Further Research}
\label{sec:further-research}

We strengthened the known results on the expressive power of PC and
URC formulas and proved structural properties of PC formulas which can
have applications to their use in knowledge compilation. By the
results of Section~\ref{sec:lb}, there is no guarantee for the
existence of a reasonably sized PC formula equivalent to a given CNF, even if it
belongs to a tractable class of Horn or q-Horn formulas. On the other
hand, PC encodings are known~\cite{KS19} to be strictly stronger than
Decision-DNNF and also DNNF which are currently the strongest routinely
used compilation languages.
Not surprisingly, identifying good auxiliary variables remains as an
important hard problem critical for the construction of small CNF encodings.
However, the results of Section~\ref{sec:irred-PC-URC} and
Section~\ref{sec:characterize-PC} suggest that heuristics for searching
good PC representations or encodings with a fixed set of auxiliary variables
may exist.

The methods for compilation not introducing new auxiliary variables are
investigated in a related ongoing research using a program
pccompile~\cite{K19}. Preliminary experiments with this
program~\cite{K19a}
demonstrate that such a compilation is
frequently tractable for formulas of a few hundreds of clauses
that appear as benchmarks in knowledge compilation. Further research is needed
to better clarify the conditions under which this compilation is tractable.
The correspondence between dual rail encoding of a PC formula
and a specific Horn function presented in Section~\ref{sec:characterize-PC}
in this paper is used also in one of the algorithms implemented in~\cite{K19}
to improve efficiency.

Let us close the paper with the following questions left open for further
research. In Theorem~\ref{thm:PC-irredundant}
we have shown that if
\(\varphi_1(\vek{x})\) and \(\varphi_2(\vek{x})\) are two
PC-irredundant formulas representing the same function \(f(\vek{x})\)
on \(n=|\vek{x}|\) variables, then \(|\varphi_2|\leq
   n^2|\varphi_1|\). It is natural to ask if the bound can be
strengthened in the following sense.

\begin{question}
Is it possible to strengthen the bound from Theorem~\ref{thm:PC-irredundant}
to $|\phi_2| = O(n |\phi_1|)$?
\end{question}

By Theorem~\ref{thm:qh-enc:main} there is
a polynomial size URC encoding for an arbitrary q-Horn formula. It is natural to
ask whether we can in fact construct a PC encoding of polynomial size.
This question is open already for the class of Horn formulas contained in
the class of q-Horn formulas and we can thus pose the following question.

\begin{question}
Let $\phi(\vek{x})$ be a Horn formula or, more generally, a URC formula.
Is there a PC encoding $\psi(\vek{x}, \vek{y})$ of $\phi$ of size polynomial
in the size of $\phi$?
\end{question}

Using the notation from Theorem~\ref{thm:qh-enc:main}, the size of the
URC encoding constructed for a q-Horn formula $\phi(\vek{x}_1, \vek{x}_2)$
is $O(|\phi| + |\vek{x}_2|^3)$.

\begin{question} \label{question:q-horn:URC-size}
Is there a URC encoding for a q-Horn formula $\phi$ of size
$O(|\phi| + |\vek{x}_2|^c)$, where $c < 3$?
\end{question}

\section*{Acknowledgements}

Both authors gratefully acknowledge the
support by Grant Agency of the Czech Republic (grant No.~GA19--19463S).

\bibliography{ms}

\begin{thebibliography}{10}

\bibitem{AGMES16}
Ignasi Ab{\'i}o, Graeme Gange, Valentin Mayer-Eichberger, and Peter~J. Stuckey.
\newblock On {CNF} encodings of decision diagrams.
\newblock In Claude-Guy Quimper, editor, {\em Integration of AI and OR
  Techniques in Constraint Programming}, pages 1--17, Cham, 2016. Springer
  International Publishing.

\bibitem{APT79}
Bengt Aspvall, Michael~F. Plass, and Robert~Endre Tarjan.
\newblock A linear-time algorithm for testing the truth of certain quantified
  boolean formulas.
\newblock {\em Information Processing Letters}, 8(3):121 -- 123, 1979.

\bibitem{BBCGK13}
Martin Babka, Tom{\'a}{\v s} Balyo, Ond{\v r}ej {\v C}epek, {\v S}tefan
  Gursk{\'y}, Petr Ku{\v c}era, and V{\'a}clav Vl{\v c}ek.
\newblock Complexity issues related to propagation completeness.
\newblock {\em Artificial Intelligence}, 203(0):19 -- 34, 2013.

\bibitem{B07}
Fahiem Bacchus.
\newblock {GAC} via unit propagation.
\newblock In Christian Bessi{\`e}re, editor, {\em Principles and Practice of
  Constraint Programming -- CP 2007}, volume 4741 of {\em Lecture Notes in
  Computer Science}, pages 133--147. Springer Berlin Heidelberg, 2007.

\bibitem{BKNW09}
Christian Bessiere, George Katsirelos, Nina Narodytska, and Toby Walsh.
\newblock Circuit complexity and decompositions of global constraints.
\newblock In {\em Proceedings of the Twenty-First International Joint
  Conference on Artificial Intelligence (IJCAI-09)}, pages 412--418, 2009.

\bibitem{BHMW09}
A.~Biere, M.~Heule, H.~van Maaren, and T.~Walsh.
\newblock {\em Handbook of Satisfiability}, volume 185 of {\em Frontiers in
  Artificial Intelligence and Applications}.
\newblock IOS Press, Amsterdam, The Netherlands, 2009.

\bibitem{BBIMM18}
Maria~Luisa Bonet, Sam Buss, Alexey Ignatiev, Joao Marques-Silva, and Antonio
  Morgado.
\newblock {MaxSAT} resolution with the dual rail encoding.
\newblock In {\em Thirty-Second AAAI Conference on Artificial Intelligence},
  2018.

\bibitem{BJM12}
Lucas Bordeaux, Mikol\'{a}\v{s} Janota, Joao Marques-Silva, and Pierre Marquis.
\newblock On unit-refutation complete formulae with existentially quantified
  variables.
\newblock In {\em Proceedings of the Thirteenth International Conference on
  Principles of Knowledge Representation and Reasoning}, KR'12, pages 75--84.
  AAAI Press, 2012.

\bibitem{BM12}
Lucas Bordeaux and Joao Marques-Silva.
\newblock Knowledge compilation with empowerment.
\newblock In M{\'a}ria Bielikov{\'a}, Gerhard Friedrich, Georg Gottlob, Stefan
  Katzenbeisser, and Gy{\"o}rgy Tur{\'a}n, editors, {\em SOFSEM 2012: Theory
  and Practice of Computer Science}, volume 7147 of {\em Lecture Notes in
  Computer Science}, pages 612--624. Springer Berlin / Heidelberg, 2012.

\bibitem{BCH90}
E.~Boros, Y.~Crama, and P.~L. Hammer.
\newblock Polynomial-time inference of all valid implications for {Horn} and
  related formulae.
\newblock {\em Annals of Mathematics and Artificial Intelligence}, 1(1):21--32,
  09 1990.

\bibitem{BHS94}
Endre Boros, Peter~L. Hammer, and Xiaorong Sun.
\newblock Recognition of q-{Horn} formulae in linear time.
\newblock {\em Discrete Applied Mathematics}, 55(1):1 -- 13, 1994.

\bibitem{BBBCS87}
R.~E. Bryant, D.~Beatty, K.~Brace, K.~Cho, and T.~Sheffler.
\newblock {COSMOS}: A compiled simulator for {MOS} circuits.
\newblock In {\em Proceedings of the 24th ACM/IEEE Design Automation
  Conference}, DAC ’87, page 9–16, New York, NY, USA, 1987. Association for
  Computing Machinery.

\bibitem{CKV12}
Ond{\v r}ej {\v C}epek, Petr Ku{\v c}era, and V{\'a}clav Vl{\v c}ek.
\newblock Properties of {SLUR} formulae.
\newblock In M{\'a}ria Bielikov{\'a}, Gerhard Friedrich, Georg Gottlob, Stefan
  Katzenbeisser, and Gy{\"o}rgy Tur{\'a}n, editors, {\em SOFSEM 2012: Theory
  and Practice of Computer Science}, volume 7147 of {\em LNCS}, pages 177--189.
  Springer Berlin Heidelberg, 2012.

\bibitem{CH11}
Y.~Crama and P.L. Hammer.
\newblock {\em Boolean Functions: Theory, Algorithms, and Applications}.
\newblock Encyclopedia of Mathematics and Its Applications. Cambridge
  University Press, 2011.

\bibitem{DM02}
Adnan Darwiche and Pierre Marquis.
\newblock A knowledge compilation map.
\newblock {\em Journal of Artificial Intelligence Research}, 17:229--264, 2002.

\bibitem{V94}
Alvaro del Val.
\newblock Tractable databases: How to make propositional unit resolution
  complete through compilation.
\newblock In {\em Knowledge Representation and Reasoning}, pages 551--561,
  1994.

\bibitem{GKPS95}
Goran Gogic, Henry Kautz, Christos Papadimitriou, and Bart Selman.
\newblock The comparative linguistics of knowledge representation.
\newblock In {\em Proceedings of the 14th International Joint Conference on
  Artificial Intelligence - Volume 1}, IJCAI'95, pages 862--869, San Francisco,
  CA, USA, 1995. Morgan Kaufmann Publishers Inc.

\bibitem{GK13}
Matthew Gwynne and Oliver Kullmann.
\newblock Generalising and unifying {SLUR} and unit-refutation completeness.
\newblock In Peter van Emde~Boas, Frans C.~A. Groen, Giuseppe~F. Italiano,
  Jerzy Nawrocki, and Harald Sack, editors, {\em SOFSEM 2013: Theory and
  Practice of Computer Science}, pages 220--232, Berlin, Heidelberg, 2013.
  Springer Berlin Heidelberg.

\bibitem{HK93}
P.L. Hammer and A.~Kogan.
\newblock Optimal compression of propositional {Horn} knowledge bases:
  Complexity and approximation.
\newblock {\em Artificial Intelligence}, 64:131~--~145, 1993.

\bibitem{IMM17}
Alexey Ignatiev, Antonio Morgado, and Joao Marques-Silva.
\newblock On tackling the limits of resolution in {SAT} solving.
\newblock In Serge Gaspers and Toby Walsh, editors, {\em Theory and
  Applications of Satisfiability Testing -- SAT 2017}, pages 164--183, Cham,
  2017. Springer International Publishing.

\bibitem{K19}
Petr Ku{\v c}era.
\newblock Program pccompile.
\newblock \url{http://ktiml.mff.cuni.cz/~kucerap/pccompile}, 2019.
\newblock Accessed: 2020-03-11.

\bibitem{K19a}
Petr Ku{\v c}era.
\newblock Tool presentation: pccompile.
\newblock Presented at the KOCOON Workshop in Arras, 2019.
\newblock \url{http://kocoon.gforge.inria.fr/slides/pccompile.pdf}.

\bibitem{KS19}
Petr Ku{\v c}era and Petr Savick{\'y}.
\newblock Propagation complete encodings of smooth {DNNF} theories, 2019,
  arXiv: 1909.06673.

\bibitem{K00}
Oliver Kullmann.
\newblock Investigations on autark assignments.
\newblock {\em Discrete Applied Mathematics}, 107(1--3):99 -- 137, 2000.

\bibitem{MFSO97}
V.~M. Manquinho, P.~F. Flores, J.~P.~M. Silva, and A.~L. Oliveira.
\newblock Prime implicant computation using satisfiability algorithms.
\newblock In {\em Proceedings Ninth IEEE International Conference on Tools with
  Artificial Intelligence}, pages 232--239, Nov 1997.

\bibitem{MIBMB19}
Antonio Morgado, Alexey Ignatiev, Maria~Luisa Bonet, Joao Marques-Silva, and
  Sam Buss.
\newblock {DRMaxSAT} with {MaxHS}: First contact.
\newblock In {\em International Conference on Theory and Applications of
  Satisfiability Testing}, pages 239--249. Springer, 2019.

\bibitem{AFSS95}
John~S. Schlipf, Fred~S. Annexstein, John~V. Franco, and R.~P. Swaminathan.
\newblock On finding solutions for extended {Horn} formulas.
\newblock {\em Inf. Process. Lett.}, 54(3):133--137, 1995.

\end{thebibliography}
\bibliographystyle{hplain}

\end{document}